\documentclass{article}

\usepackage{amsmath}
\usepackage{amssymb}
\usepackage{fdsymbol}
\usepackage{amsthm}
\usepackage{tikz}
\usetikzlibrary{calc,intersections}

\usepackage{listings}

\newtheorem{theorem}{Theorem}[section]

\newtheorem{lemma}[theorem]{Lemma}

\newcommand{\ms}[1]{{\mathsf{#1}}}
\newcommand{\mb}[1]{{\mathbb{#1}}}
\newcommand{\od}[2]{{#1}\;\measuredrightangle\;{#2}}

\title{Zig-zagging in a Triangulation}
\author{Wouter Kuijper \and Victor Ermolaev \and Harm Berntsen}

\begin{document}

\maketitle

\section{Introduction}
We present an oblivious walk for point-location in 2-dimensional
triangulations and a corresponding, strictly monotonically decreasing
distance measure.

\section{Problem Description}

Let $T = (V, F, E)$ be a 2-dimensional triangular tiling of $\mb{R}^2$
in quad-edge representation~\cite{guibas1985primitives}, consisting of
a set of vertices $V$, a set of faces $F$ and a set of half-edges
$E$. We say $T$ is \emph{locally finite} iff for any circle in
$\mb{R}^2$ it holds that the set consisting of all edges (vertices,
faces) in $T$ that intersect the circle (circumference or interior) is
finite.

Let $e \in E$ be some half-edge of this mesh starting at vertex $e_1$
and ending at vertex $e_2$. With $\ms{inv}(e)$ we denote the twin of
$e$ such that $\ms{inv}(e)_1 = e_2$ and $\ms{inv}(e)_2 = e_1$. With
$\ms{face}(e)$ we denote the face $f \in F$ that is, by convention, to
the left of the half-edge $e$. With $\ms{next}(e)$ we denote the next
half-edge (in the face winding order) that is also a side of
$\ms{face}(e)$, i.e.: $\ms{face}(\ms{next}(e)) = \ms{face}(e)$. With
$\ms{prev}(e)$ we denote the previous half-edge (in the face
winding order) that is also a side of $\ms{face}(e)$, i.e.:
$\ms{face}(\ms{prev}(e)) = \ms{face}(e)$. Because we are dealing with
a triangular mesh, in particular, it holds: $\ms{next}(\ms{next}(e)) =
\ms{prev}(e)$ and vice versa.

The problem of point-location~\cite{devillers2002walking} that we
consider in this note can be formulated as follows: Given a point $p
\in \mb{R}^2$ and an initial half-edge $e_\ms{init} \in E$ find some
half-edge $e_\ms{goal} \in E$, using only $\ms{next}(\cdot)$,
$\ms{prev}(\cdot)$ and $\ms{inv}(\cdot)$ operations, such that $p \in
\ms{face}(e_\ms{goal})$, i.e.: the point is on the face. We,
additionally, require the algorithm to be \emph{oblivious}, meaning
that the next chosen half-edge only depends on the last-visited
half-edge and the goal location.

\section{Example of a Zig-Zagging Walk}

A walk in a triangulation is a sequence of faces. We concern ourselves
here only with \emph{directed} walks, such walks attempt to close the
distance between some starting edge (or its corresponding face) and a
target point. Directed walks are important in practice as they are
frequently used to navigate meshes as efficient, flexible
space-partitioning data-structures.

One possible application of a directed walk would be to translate
mouse-clicks to selected faces in a CAD-application. Because such type
of interaction exhibits a lot of locality ---i.e.: clicks often happen
in close proximity to each other--- it is much more efficient to walk
to the corresponding face from the last-visited one then to perform a,
potentially exhaustive, sweep of the entire mesh, which, in
representative cases, may measure in the order of millions or even
billions of faces.

Before we turn to the exposition of the zig-zagging walk, its
corresponding distance measure and termination proof, let us first
investigate briefly why a zig-zagging motion is a useful and intuitive
analogy in thinking about point-location in two dimensional
triangulations.

Since an oblivious walk ``forgets'' its original point of departure as
well as the route followed to get to the current face, it suffices to
consider only the current half-edge and the corresponding current
face (which, by our convention, is always to the left of the current
half-edge) when thinking about the point-location problem.

\begin{figure}\center
  \begin{tikzpicture}
    \node (e1) at (-5,-3) [inner sep=0] {$\circ$};
    \node (e1l) at (e1) [anchor=north] {$e_1$};
    \node (e2) at (5,-3) [inner sep=0] {$\circ$};
    \node (e2l) at (e2) [anchor=north] {$e_2$};
    \node (c) at (0,0) [inner sep=0] {$\circ$};
    \node (t1) at (-5,3) [inner sep=0] {$\circ$};
    \node (t2) at (5,3) [inner sep=0] {$\circ$};
    \node (f) at (0,-2) [inner sep=1mm] {$f$};
    \node (fl) at (-4,0) [inner sep=1mm] {$f_l$};
    \node (fr) at (4,0) [inner sep=1mm] {$f_r$};
    \node (p) at (0,2) [inner sep=0] {$\circ$};
    \node (pl) at (p) [anchor=south,yshift=0mm] {$p$};
    \path [draw] (e1) -- node [above] (e) {$e$} (e2);
    \path [draw] (e1) -- (c);
    \path [draw] (e2) -- (c);
    \path [draw] (e1) -- (t1);
    \path [draw] (e2) -- (t2);
    \path [draw] (c) -- (t1);
    \path [draw] (c) -- (t2);
    \path [draw] (t1) -- (t2);
    \path [draw,dashed,->] (f) -- (fl);
    \path [draw,dashed,->] (f) -- (fr);
    \path [draw,dashed,->] (fl) -- (p);
    \path [draw,dashed,->] (fr) -- (p);
  \end{tikzpicture}
  \caption{An example of a zig-zagging motion exhibited during a
    walk.}\label{fig:zig-zag}
\end{figure}
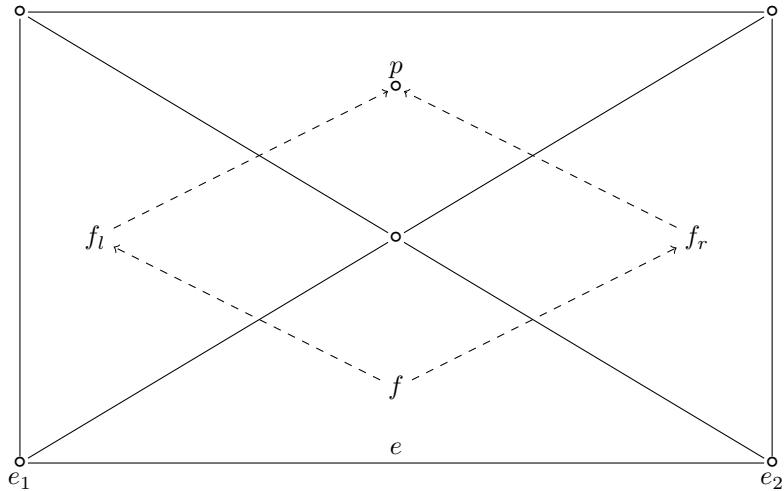

As an example consider Figure~\ref{fig:zig-zag}. Assume that at some
point during the walk we arrive at half-edge $e$ and corresponding
face $f = \ms{face}(e)$. The target point $p$ is ``straight ahead'' in
the direction of the tip of $f$. However, going straight is not an
option, because the only faces that are directly linked to $f$ (in the
quad-edge representation) are $f_l =
\ms{face}(\ms{inv}(\ms{prev}(e)))$ and $f_r =
\ms{face}(\ms{inv}(\ms{next}(e)))$. Since ``turning back'' is not an
option either, any route to the target point $p$ will have to go
either over the left or over the right, hence, any route will make a
``zig-zag'' motion to reach the target.

The previous example may seem particularly inefficient in terms of
Euclidian distance traveled, however note that the complexity of the
walk does not, in any way, depend on the Euclidian distance, rather it
depends on the \emph{link distance} which, in our case, could be
defined as the number of atomic quad-edge relations ($\ms{next}$,
$\ms{prev}$ and $\ms{inv}$) that had to be traversed.

One problem that we quickly run into when developing a
walking-algorithm is the fact that Euclidian distance is not a
sufficient measure to prove termination: steps do not always strictly
decrease the Euclidian distance to the target. Therefore, to obtain a
clear proof of termination, we need an alternative distance measure
that takes into account the fact that a step in the walk can make
progress not only in terms of Euclidian distance to the target but
also in terms of \emph{orientation towards the target}.

\section{Distance Measure}
%
For some half-edge $e$ and point $p$ such that $e$ does not intersect
$p$ we define the \emph{oriented distance} of $e$ to $p$ denoted
$\od{e}{p}$ as a pair $[d, \alpha]$ where $d$ is defined to be the
Euclidian distance of point $p$ to the closest point $e_p$ on the
(finite!) line segment $e_1;e_2$ corresponding to half-edge $e$ and
$\alpha$ is defined to be the (smallest) angle between $e_1;e_2$ and
the ray $p;e_p$ or zero in case $p = e_p$.
Since the shortest ray from $p$ to $e$ ends at either $e_1$, $e_2$ or
the orthogonal projection of $p$ on $e_1;e_2$ it follows $\alpha \ge
\pi/2$.
For some examples of how this oriented distance measure is defined
under various relative configurations of the current edge and the
target point, see Figure~\ref{fig:dist1}.

\begin{figure}\center
  \begin{tikzpicture}[scale=.9]
  	\node (e1) at (-2, 0) [inner sep=0cm] {$\circ$};
	\node (e2) at (2, 0) [inner sep=0cm] {$\circ$};
	\path [draw] (e1) -- (e2); 
	\node (e1label) [anchor=north] at (e1) {$e_1$};
  	\node (e2label) [anchor=north] at (e2) {$e_2$};
	
	\node (p1) at (-1,3) [inner sep=0cm] {$\circ$};
  	\node (ep) at (-1,0) [inner sep=0cm] {$\circ$};
	\path [draw, dashed,<->] (p1) -- node [right] {$d_1$} (ep); 	
  	\node (p1label) [anchor=south] at (p1) {$p_1$};
  	\node (ep1label) [anchor=north] at (ep) {$e_p$};
  	
  	\path [draw] (-.6,0) arc (0:90:.4cm);
  	\node (alpha1) at (-.5,.5) {$\alpha_1$};
  	
  	\node (p2) at (-6, 3) [inner sep=0cm] {$\circ$};
	\path [draw, dashed,<->] (p2) -- node [right] {$d_2$} (e1); 	
  	\node (p2label) [anchor=south] at (p2) {$p_2$};  	

  	\path [draw] (-1.7, 0) arc (0:135:.4cm);
  	\node (alpha2) at (-1.6,.5) {$\alpha_2$};
  	
  	\node (p3) at (6, 3) [inner sep=0cm] {$\circ$};
	\path [draw, dashed,<->] (p3) -- node [right] {$d_3$} (e2); 	
  	\node (p3label) [anchor=south] at (p3) {$p_3$};  
  	
  	\path [draw] (1.7, 0) arc (180:45:.4cm);
  	\node (alpha3) at (1.6,.5) {$\alpha_3$};	
  \end{tikzpicture}
  \caption{Examples of point--half-edge distance for various locations of a point.}
  \label{fig:dist1}
\end{figure}
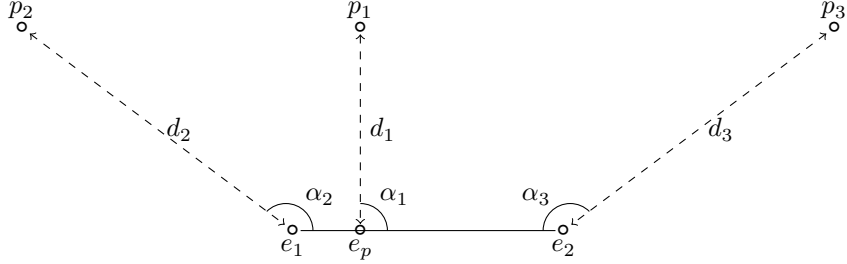

We now define a strict partial ordering $<$ on oriented distances such
that
$$[d, \alpha] < [d', \alpha']\text{ iff }d < d' \lor (d = d' \land
\alpha < \alpha')$$
In order to see how this distance measure allows us to compare various
candidate successor edges consider the example in
Figure~\ref{fig:areas}. Here we show a classification of the points in
the space that is to the left of the line supporting the current
half-edge $e$ (but not \emph{on} the current face). So let $p \in
\mb{R}^2$ be the target point and let $l = \ms{inv}(\ms{next}(e))$ and
$r = \ms{inv}(\ms{prev}(e))$ be the two possible successor half-edges,
now let $[d_l, \alpha_l] = \od{l}{p}$ and $[d_r, \alpha_r] =
\od{r}{p}$, we classify the space where the target point $p$ may
reside as follows:
\begin{itemize}
\item $d_l < d_r$ and the point of $l$ closest to $p$ is $l_1$ (I$_l$)
\item $d_l < d_r$ and the point of $l$ closest to $p$ is strictly
  in-between $l_1$ and $l_2$ (II$_l$)
\item $d_l = d_r$ and $\alpha_l < \alpha_r$ (III$_l$)
\item $d_l = d_r$ and $\alpha_l = \alpha_r$ (IV)
\item $d_l > d_r$ and $\alpha_l > \alpha_r$ (III$_r$)
\item $d_l > d_r$ and the point of $r$ closest to $p$ is strictly in-between $r_1$ and $r_2$ (II$_r$)
\item $d_l > d_r$ and the point of $r$ closest to $p$ is $r_2$ (I$_r$)
\end{itemize}
Note that area I$_l$ or I$_r$ becomes a line in case the triangle is
right in its left or right base angle respectively and empty in case
the corresponding angle becomes obtuse. All the other areas remain
non-empty no matter how obtuse or how thin the triangle becomes.

\begin{figure}\center
  \begin{tikzpicture}[scale=1]
    \node (e1) [inner sep=0cm] at (-2, 0) {$\circ$};
    \node (e2) [inner sep=0cm] at (2, 0) {$\circ$};
    \coordinate (ee1) at (-5, 0);
    \coordinate (ee2) at (5, 0);
    \node (a) [inner sep=0cm] at (0,3) {$\circ$};
    \node (labelA) [anchor=north] at (e1) {$l_1 = e_1$};
    \node (labelB) [anchor=north] at (e2) {$e_2 = r_2$};
    \node (labelCl) [anchor=north east,xshift=-.3cm] at (a) {$l_2$};
    \node (labelCe) [anchor=north,yshift=-.1cm] at (a) {$=$};
    \node (labelCr) [anchor=north west,yshift=-.1cm,xshift=+.3cm] at (a) {$r_1$};
    \path[draw,dashed] (ee1) -- (e1);
    \path[draw] (e1) -- node [above] {$e$} (e2);
    \path[draw,dashed] (e2) -- (ee2);
    \path[draw] (e1) -- node [above,xshift=-.2cm] {$l$} (a);
    \path[draw] (e2) -- node [above,xshift=.2cm] {$r$} (a);
    \path[draw,dashed] (a) -- ++(3,2);
    \path[draw,dashed] (e2) -- ++(3,2);
    \path[draw,dashed] (a) -- ++(-3,2);
    \path[draw,dashed] (e1) -- ++(-3,2);
    \node (iv) at (0,5.5) {IV};
    \path [draw,dashed] (a) -- (iv);
    \path [draw,dashed] (iv) -- ++(0,1);
    \node (i) at (4.5, .8) {I$_r$};
    \node (ip) at (-4.5, .8) {I$_l$};
    \node (ii) at (3, 3) {II$_r$};
    \node (iip) at (-3, 3) {II$_l$};
    \node (iii) at (-1, 5) {III$_l$};
    \node (iii) at (1, 5) {III$_r$};
  \end{tikzpicture}
  \caption{The various areas where the target point may reside.}\label{fig:areas}
\end{figure}
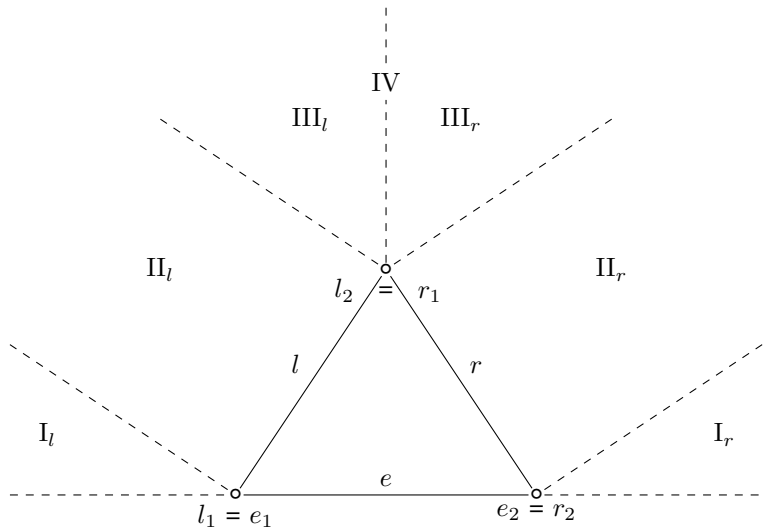

\section{Oriented Distance Walk}
We are now in a position to describe the algorithm properly. Let $e \in E$ be some initial half-edge,
and let $p \in \mb{R}^2$ be the goal location. The algorithm then goes as follows:

\lstdefinelanguage{alg}{}

\begin{lstlisting}[language=alg, mathescape ,numbers=left]
    $\text{\bf if}$ $p\ \text{is to the right--of}\ e$ $\text{\bf then}$ $e \leftarrow \ms{inv}(e)$
    $\text{\bf while}$ $p \notin \ms{face}(e)$
        $l, r \leftarrow \ms{inv}(\ms{next}(e)), \ms{inv}(\ms{prev}(e))$
        $\text{\bf if}$ $(\od{l}{p}) < (\od{r}{p})$ $\text{\bf then}$ $e \leftarrow l$ $\text{\bf else}$ $e \leftarrow r$
\end{lstlisting}

In line 1 we bootstrap the algorithm by ensuring the invariant that
$p$ is always to the left of $e$. In line 2 we check the termination
condition that $p$ is on the current face (i.e.: the face to the
left of the current half-edge). In line 3 we compute the two possible
successor edges. In line 4 we pick the successor edge that has the
smallest oriented distance to the target point.

Note that, in this algorithm, ties are broken by defaulting to the
right-successor edge. Clearly, it is possible to have a leftmost
version by defaulting to the left instead, or a non-deterministic
version by allowing either the left or the right-successor edge
whenever $\od{l}{p}$ is not strictly less or greater-than
$\od{r}{p}$. For the remainder we will assume the non-deterministic
version of the algorithm, all the results that follow will then hold
also for (both) deterministic versions, as a corollary.

We shall first prove a number of properties of the walk which
consequently will allow us to prove the main result, i.e.: guaranteed
termination of the zig-zagging walk. First we show that the point
will, as an invariant, always be to the left of the current half-edge.

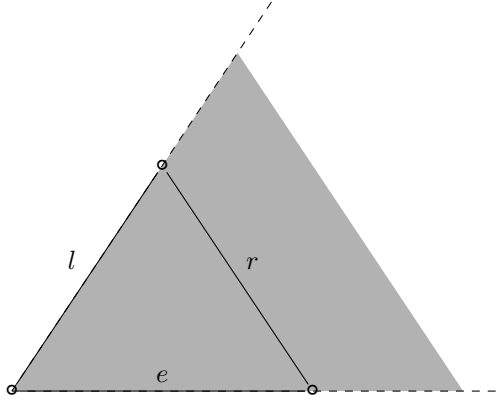
\begin{figure}\center
  \begin{tikzpicture}[scale=1]  
    \fill[black!30] (-2,0) -- (4,0) -- (1,4.5) -- (-2,0);
    \draw[dashed] (-2,0) -- (4.5,0);
    \draw[dashed] (-2,0) -- (1.5,5.25);
    \node (e1) [inner sep=0cm] at (-2, 0) {$\circ$};
    \node (e2) [inner sep=0cm] at (2, 0) {$\circ$};
    \node (tip) [inner sep=0cm] at (0,3) {$\circ$};
    \path[draw] (e1) -- node [above] {$e$} (e2);
    \path[draw] (e1) -- node [above, xshift=-.2cm] {$l$} (tip);
    \path[draw] (e2) -- node [above, xshift=.2cm] {$r$} (tip);
  \end{tikzpicture}
  \caption{Lemma~\ref{lem:halfspace}, step 1: $p$ lies somewhere in the gray cone}\label{fig:lemma:1}
\end{figure}

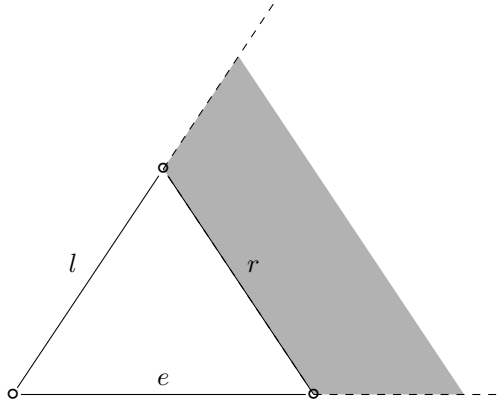
\begin{figure}\center
  \begin{tikzpicture}[scale=1]
    \fill[black!30] (2,0) -- (4,0) -- (1,4.5) -- (0,3) -- (2,0);
    \draw[dashed] (2,0) -- (4.5,0);
    \draw[dashed] (2,0) -- (0,3);
    \draw[dashed] (0,3) -- (1.5,5.25);
    \node (e1) [inner sep=0cm] at (-2, 0) {$\circ$};
    \node (e2) [inner sep=0cm] at (2, 0) {$\circ$};
    \node (tip) [inner sep=0cm] at (0,3) {$\circ$};
    \path[draw] (e1) -- node [above] {$e$} (e2);
    \path[draw] (e1) -- node [above, xshift=-.2cm] {$l$} (tip);
    \path[draw] (e2) -- node [above, xshift=.2cm] {$r$} (tip);
  \end{tikzpicture}
  \caption{Lemma~\ref{lem:halfspace}, step 2: $p$ lies somewhere in the gray, truncated cone}\label{fig:lemma:2}
\end{figure}

\begin{figure}\center
  \begin{tikzpicture}[scale=1]
    \fill[black!30] (0,3) -- (1,4.5) -- (2,4.33333) -- (0,3);
    \draw[dashed] (0,3) -- (1.5,5.25);
    \draw[dashed] (0,3) -- (3,5);
    \node (e1) [inner sep=0cm] at (-2, 0) {$\circ$};
    \node (e2) [inner sep=0cm] at (2, 0) {$\circ$};
    \node (tip) [inner sep=0cm] at (0,3) {$\circ$};
    \path[draw] (e1) -- node [above] {$e$} (e2);
    \path[draw] (e1) -- node [above, xshift=-.2cm] {$l$} (tip);
    \path[draw] (e2) -- node [above, xshift=.2cm] {$r$} (tip);
    \draw (-0.395,2.4) arc (56.6:123.4:-.72cm);
    \node (alr) at (0,2) {$\angle lr$};
  \end{tikzpicture}$\quad$
  \caption{Lemma~\ref{lem:halfspace}, step 3: $p$ lies somewhere in the gray cone}\label{fig:lemma:3}
\end{figure}

\begin{figure}\center
  \begin{tikzpicture}[scale=1]  
    \fill[black!30] (-2,0) -- (0,3) -- (2,0) -- (4,0) -- (0,4.5) -- (-4,0) -- (-2,0);
    \draw[dashed] (-2,0) -- (-4.5,0);
    \draw[dashed] (2,0) -- (4.5,0);
    \node (e1) [inner sep=0cm] at (-2, 0) {$\circ$};
    \node (e2) [inner sep=0cm] at (2, 0) {$\circ$};
    \node (tip) [inner sep=0cm] at (0,3) {$\circ$};
    \path[draw] (e1) -- node [above] {$e$} (e2);
    \path[draw] (e1) -- node [above, xshift=-.2cm] {$l$} (tip);
    \path[draw] (e2) -- node [above, xshift=.2cm] {$r$} (tip);
  \end{tikzpicture}
  \caption{Lemma~\ref{lem:mono}, step 1: $p$ lies somewhere in the gray zone}\label{fig:lemma:mono:1}
\end{figure}

\begin{figure}\center
  \begin{tikzpicture}[scale=1]  
    \fill[black!30] (-2,0) -- (0,3) -- (2,0) -- (2,4.5) -- (-2,4.5) -- (-2,0);
    \draw[dashed] (-2,0) -- (-2,5);
    \draw[dashed] (2,0) -- (2,5);
    \node (e1) [inner sep=0cm] at (-2, 0) {$\circ$};
    \node (e2) [inner sep=0cm] at (2, 0) {$\circ$};
    \node (tip) [inner sep=0cm] at (0,3) {$\circ$};
    \path[draw] (e1) -- node [above] {$e$} (e2);
    \path[draw] (e1) -- node [above, xshift=-.2cm] {$l$} (tip);
    \path[draw] (e2) -- node [above, xshift=.2cm] {$r$} (tip);
  \end{tikzpicture}
  \caption{Lemma~\ref{lem:mono}, step 2a: $p$ lies somewhere in the gray zone}\label{fig:lemma:mono:2}
\end{figure}

\begin{figure}\center
  \begin{tikzpicture}[scale=1]
    \fill[black!30] (-2,0) -- (-4.5,0) -- (-4,1.3333) -- (-2,0);
    \draw[dashed] (-2,0) -- (-5,2);
    \draw[dashed] (-2,0) -- (-5,0);
    \node (e1) [inner sep=0cm] at (-2, 0) {$\circ$};
    \node (e2) [inner sep=0cm] at (2, 0) {$\circ$};
    \node (tip) [inner sep=0cm] at (0,3) {$\circ$};
    \path[draw] (e1) -- node [above] {$e$} (e2);
    \path[draw] (e1) -- node [above, xshift=-.2cm] {$l$} (tip);
    \path[draw] (e2) -- node [above, xshift=.2cm] {$r$} (tip);
    \draw (-1.4,0) arc (0:56:.6cm);
    \node (ael) at (-1,.5) {$\angle el$};
  \end{tikzpicture}$\quad$
  \caption{Lemma~\ref{lem:mono}, step 2b: $p$ lies somewhere in the gray cone}\label{fig:lemma:mono:3}
\end{figure}

\begin{lemma}[Target in Half-space]\label{lem:halfspace}
At every iteration of the walk it holds that $p$ lies in the half-space
to the left of $e$.
\end{lemma}

\begin{proof}
By induction on the number of iterations. The basis is ensured as a
postcondition of line 1. For the inductive step assume, w.l.o.g., that
the left successor edge $l$ is chosen in line 4, which implies, in the
non-deterministic version, $\od{l}{p} \le \od{r}{p}$.  Now, for
contradiction, assume that $p$ lies strictly to the right of $l$. By
inductive hypothesis we have that $p$ is to the left of $e$ which
places it in the cone formed by $l$ and $e$
(cf. Figure~\ref{fig:lemma:1}). In addition we have the loop invariant
$p \notin \ms{face}(e)$ which places $p$ further inside the now
truncated cone formed by $l$, $e$ and $r$
(cf. Figure~\ref{fig:lemma:2}). Now let $[d_l, \alpha_l] = \od{l}{p}$
and $[d_r, \alpha_r] = \od{r}{p}$, and note that a ray cast from any
point in the truncated cone to $l$ must pass through $r$ hence we
obtain: $d_l \ge d_r$. Now, by assumption, $\od{l}{p} \le \od{r}{p}$
which implies $d_l \le d_r$. From the latter two facts it follows that
$d_l = d_r$ and $\alpha_l \le \alpha_r$ which places $p$ still further
inside the cone formed by $l$ and the line orthogonal to $r$ starting
at the tip vertex (cf. Figure~\ref{fig:lemma:3}). However, we see from
the diagram that this implies $\alpha_l = \alpha_r + \angle lr$ which
directly contradicts our assumption that $\alpha_l \le \alpha_r$.
\end{proof}

\begin{lemma}[Monotonicity]\label{lem:mono}
  With every step of the zig-zagging walk the oriented distance of the
  current edge to the target strictly decreases.
\end{lemma}
\begin{proof}
  Assume for contradiction that the oriented distance does not
  strictly decrease, i.e.: $(\od{e}{p}) \le (\od{l}{p})$ and
  $(\od{e}{p}) \le (\od{r}{p})$.  First note that by
  Lemma~\ref{lem:halfspace} we have that $p$ is to the left of $e$,
  and by the loop invariant we have that $p$ is outside of the current
  face (cf. Figure~\ref{fig:lemma:mono:1}). Now let $e_p$ be the point
  on $e$ closest to $p$. We make a case distinction. As a first case
  consider $e_p \ne e_1$ and $e_p \ne e_2$, i.e.: $e_p$ is properly
  contained between $e_1$ and $e_2$. In this case the projection of
  $p$ on $e$ must have been orthogonal
  (cf. Figure~\ref{fig:lemma:mono:2}). However this means that the ray
  from $p$ to $e$ must have passed through one of the two sides $r$ or
  $l$ and through the interior of the face which contradicts our
  assumption that $(\od{e}{p}) \le (\od{l}{p})$ and $(\od{e}{p}) \le
  (\od{r}{p})$. For the second case consider $e_p = e_1$ or $e_p =
  e_2$. So, w.l.o.g, assume $e_p = e_1$.  Now let $[d_l, \alpha_l] =
  \od{l}{p}$ and $[d_e, \alpha_e] = \od{e}{p}$. Since $e_1 = l_1$ is a
  shared point with $l$ it must hold $d_l \le d_e$, and hence, by
  assumption that $(\od{e}{p}) \le (\od{l}{p})$, it must follow that
  $d_l = e_l$ and $\alpha_e \le \alpha_l$, this places $p$ in the cone
  between $e$ and the line orthogonal to $l$ emanating from $e_1$
  (cf. Figure~\ref{fig:lemma:mono:3}). However, we see from the
  diagram that this implies $\alpha_e = \alpha_l + \angle le$, which
  directly contradicts our assumption that $\alpha_e \le \alpha_l$.
\end{proof}

We are now in position to formulate and prove finite time termination
theorem. For this we need local finiteness of $T$ to rule out the
possibility of asymptotic walks in dense tessellations.

\begin{theorem}
If $T$ is locally finite, the zig-zagging walk always terminates.
\end{theorem}
\begin{proof}
  Let $p \in \mb{R}^2$ be the target point and $e \in E$ the initial
  half-edge, let $[d, \alpha] = \od{e}{p}$, and let $E' = \{ e' \in
  E\ |\ (\od{e'}{p}) \le (\od{e}{p}) \}$. We say $E'$ is the
  neighborhood of $e$ and $p$. Note that all edges in the neighborhood
  $E'$ intersect a disc centered at $p$ with radius $d$, hence, by
  local finiteness of $T$, it follows $E'$ is a finite set.  Now by
  Lemma~\ref{lem:mono} we have that the zig-zagging walk visits a
  sequence of half-edges in $E'$ that exhibits a strictly
  monotonically decreasing chain of oriented distances, and, since
  $E'$ is finite, this implies termination.
\end{proof}

If we define the local size of the mesh as the size of the
neighborhood as defined in the proof, we actually get the stronger
result that the zig-zagging walk always terminates in a number of
steps that is bounded by the local size of the mesh.

\subsection{Spiraling in a Tetrahedralization}

One of the obvious next questions to consider is whether this result
can be generalized from two to three dimensions. Although we defer formal
treatment to future work, we already give some intuition and a proof
sketch as to how this might be done.

\begin{figure}\center  
  \begin{tikzpicture}[scale=1]
    \node (p) at (0, 0) [inner sep=0cm] {$\circ$};
    \node (a) at (2.5, 0) [inner sep=0cm] {$\circ$};
    \node (b) at (4.5, -1.5) [inner sep=0cm] {$\circ$};
    \node (c) at (3, -2) [inner sep=0cm] {$\circ$};
    \draw (p) circle (2.5cm);
    \path [draw,dashed,<->] (p) -- (a);
    \path [draw] (a) -- (b) -- (c) -- (a);
    \draw [->] (3.5,-.75) arc (-36.39:-75.97:1.25cm);
  \end{tikzpicture}
\caption{A example of a reorienting step in two
  dimensions.}\label{fig:2da}
  \vspace{1cm}
  \begin{tikzpicture}[scale=1]
    \node (p) at (0, 0) [inner sep=0cm] {$\circ$};
    \node (a) at (2.5, 0) [inner sep=0cm] {$\circ$};
    \node (c) at (3, -2) [inner sep=0cm] {$\circ$};
    \node (b) at (1,-1) [inner sep=0cm] {$\circ$};
    \draw (a) arc (0:322:2.5cm);
    \draw (p) circle (1.4142cm);
    \path [draw,dashed,<->] (p) -- (b);
    \path [draw] (a) -- (c) -- (b) -- (a);
    \draw [->] (0,2.5) -- (0,1.4142);
    \draw [->] (-2.5,0) -- (-1.4142,0);
    \draw [->] (2.5,0) -- (1.4142,0);
    \draw [->] (0,-2.5) -- (0,-1.4142);
  \end{tikzpicture}
\caption{An example of an approaching step in two
  dimensions.}\label{fig:2dd}
\end{figure}

First observe that, in two dimensions, we can define the ``current
neighborhood circle'' to be the smallest circle surrounding the
target point that intersects the current half-edge. A zig-zagging walk
then consists of steps that change the angle of the current half-edge
to make it more tangent to the current neighborhood circle
(cf. Figure~\ref{fig:2da}) eventually followed by steps that approach
the target and thereby shrink the radius of the current neighborhood
circle (cf. Figure~\ref{fig:2dd}).

In three dimensions the analogon to the neighborhood circle would be the
smallest \emph{sphere} surrounding the target point that intersects
the face. A spiraling walk through a tetrahedralization would then
consist, analogously, of steps that change the orientation of the
current face to make it more tangent to the current neighborhood
sphere (cf. Figure~\ref{fig:3da}) eventually followed by steps that
approach the target and thereby shrink the radius of the current
neighborhood sphere (cf. Figure~\ref{fig:3dd}).

\begin{figure}\center
  \begin{tikzpicture}[scale=1]
    \node (p) at (0, 0) [inner sep=0cm] {$\circ$};
    \node (a) at (2.5, 0) [inner sep=0cm] {$\circ$};
    \node (b) at (4.5, -1.5) [inner sep=0cm] {$\circ$};
    \node (c) at (3, -2) [inner sep=0cm] {$\circ$};
    \node (d) at (3.5,-.25) [inner sep=0cm] {$\circ$};
    \draw (p) circle (2.5cm);
    \draw (a) arc (0:-180:2.5cm and 1cm);
    \draw [dotted] (a) arc (0:180:2.5cm and 1cm);
    \path [draw,dashed,<->] (p) -- (a);
    \path [draw] (a) -- (b) -- (c) -- (a);
    \path [draw] (a) -- (d) -- (b);
    \path [draw,dotted] (d) -- (c);
    \draw [->] (3.5,-.75) arc (-36.39:-75.97:1.25cm);
  \end{tikzpicture}
\caption{A example of a reorienting step in three
  dimensions.}\label{fig:3da}
  \vspace{1cm}
  \begin{tikzpicture}[scale=1]
    \node (p) at (0, 0) [inner sep=0cm] {$\circ$};
    \node (a) at (2.5, 0) [inner sep=0cm] {$\circ$};
    \node (c) at (3, -2) [inner sep=0cm] {$\circ$};
    \node (b) at (1,-1) [inner sep=0cm] {$\circ$};
    \node (d) at (3.5,-.25) [inner sep=0cm] {$\circ$};
    \draw (a) arc (0:322:2.5cm);
    \draw (p) circle (1.4142cm);
    \draw (1.4142,0) arc (0:-180:1.4142cm and .6cm);
    \draw [dotted] (1.4142,0) arc (0:180:1.4142cm and .6cm);
    \path [draw,dashed,<->] (p) -- (b);
    \path [draw] (a) -- (c) -- (b) -- (a);
    \path [draw] (a) -- (d) -- (c);
    \path [draw,dotted] (d) -- (b);
    \draw [->] (0,2.5) -- (0,1.4142);
    \draw [->] (-2.5,0) -- (-1.4142,0);
    \draw [->] (2.5,0) -- (1.4142,0);
    \draw [->] (0,-2.5) -- (0,-1.4142);
  \end{tikzpicture}
\caption{An example of an approaching step in three
  dimensions.}\label{fig:3dd}
\end{figure}

In the two dimensional case a single angle was sufficient to
characterize the orientation of the current half-edge w.r.t. the target
point. In the three dimensional case we need at least 2 angles to
characterize the orientation of the current face w.r.t. the target
point.

In particular we define the \emph{ray} as the shortest line segment
from the target point to the closest point on the face, we must then
fix a roll \emph{roll axis} as a line-segment that lies in the face
and intersects the ray.

Given such a roll axis we then define the \emph{pitch axis} as the
line that is orthogonal to both the ray as well as the roll axis.  We
further define the \emph{pitch angle} of the face as the angle between
the ray and the roll axis and we define the \emph{roll angle} as the
angle between the face and the plane defined by the pitch and roll
axis.

We then define a \emph{pitch minimizing roll axis} to be any roll axis
that minimizes the pitch angle, and we define a \emph{minimizing roll
  axis} to be any pitch minimizing roll axis that minimizes the roll
angle, i.e.: it minimizes \emph{firstly} the pitch angle and
\emph{secondly} the roll angle. We refer to the latter as the
\emph{minimal pitch angle} and the \emph{minimal roll angle}.

For some face $f$ and some point $p$ where $f$ does not contain $p$,
we can now define the oriented distance of $f$ to $p$, notation:
$\od{f}{p}$, as the triple $[d, \alpha, \beta]$ where $d$ is the
Euclidian distance of the target point to the closest point on the
face, $\alpha$ is the minimal pitch angle and $\beta$ is the minimal
roll angle. We then define a strict partial order as follows:
\begin{align*}
  [d, \alpha, \beta] < [d', \alpha', \beta']\text{ iff }
  d < d' \lor (d = d' \land \alpha < \alpha') \lor (d = d' \land \alpha = \alpha' \land \beta < \beta')
\end{align*}

\section{Discussion and Future Work}
We defer to future work the formal treatment of the three dimensional
case (for which we already gave some intuitions in the previous
section), and, possibly, a full generalization to $n$--dimensional
simplexes.

Another possible direction that we are interested in is to look for a
robust version of the zig-zagging walk that does not require infinite
precision in the evaluation of the predicates.
The hope there is that the clear termination proof in the infinite
precision case can function as a starting point in proving termination
in the finite precision case.

\bibliography{main}{} \bibliographystyle{plain}

\begin{thebibliography}{1}

\bibitem{devillers2002walking}
Olivier Devillers, Sylvain Pion, and Monique Teillaud.
\newblock Walking in a triangulation.
\newblock {\em International Journal of Foundations of Computer Science},
  13(02):181--199, 2002.

\bibitem{guibas1985primitives}
Leonidas Guibas and Jorge Stolfi.
\newblock Primitives for the manipulation of general subdivisions and the
  computation of voronoi.
\newblock {\em ACM transactions on graphics (TOG)}, 4(2):74--123, 1985.

\end{thebibliography}

\end{document}